\documentclass[11pt, a4paper]{article}

\usepackage{a4wide,url}
\usepackage{graphicx,amssymb,amsmath,color,amsthm}

\newtheorem{lemma}{Lemma}
\newtheorem{theorem}[lemma]{Theorem}

\newtheorem{observ}[lemma]{Observation}

\newcommand{\R}{\mathbb{R}}

\begin{document}
\pagestyle{plain}

\title{Covering and Piercing Disks with Two Centers}
 
\author{Hee-Kap Ahn\thanks{Department of Computer Science and Engineering, POSTECH, Pohang, Korea. {\tt \{heekap,
      helmet1981\}@postech.ac.kr}} \and Sang-Sub Kim\footnotemark[2] \and Christian Knauer\thanks{Institute of Computer Science, Universit\"at Bayreuth, 95440 Bayreuth, Germany.
  {\tt christian.knauer@uni-bayreuth.de}}
\and
Lena Schlipf\thanks{Institute of Computer Science, Freie Universit\"at Berlin, Germany. {\tt schlipf@mi.fu-berlin.de}} \and Chan-Su Shin \thanks{
Department of Digital and Information Engineering,
    Hankuk University of Foreign Studies, Yongin, Korea.
    {\tt cssin@hufs.ac.kr}}
 \and Antoine Vigneron\thanks{Geometric Modeling and Scientific Visualization Center, KAUST,
    Thuwal, Saudi Arabia. {\tt antoine.vigneron@kaust.edu.sa}}}

\date{}
\maketitle

\begin{abstract}
  We give exact and approximation algorithms
  for two-center problems when the input is a set
  $\mathcal{D}$ of disks in the plane. We first
  study the problem of finding two
  smallest congruent disks such that each disk in $\mathcal{D}$ 
  intersects one of these two disks. Then we study the problem of
  covering the set $\mathcal{D}$ by two smallest congruent disks.
\end{abstract}

\section{Introduction}

The standard \emph{two-center problem} is a well known and extensively studied problem:
Given a set $P$ of $n$ points in the plane, find two smallest congruent disks that cover all points in $P$. 
The best known deterministic algorithm runs in $O(n \log^2 n \log^2 \log n)$~\cite{chan99}
and there is a randomized algorithm with expected running time $O(n\log^2n)$ \cite{eppstein97}.
There has also been a fair amount of work on several variations of the two-center problem:
for instance, the two-center problem for weighted points \cite{drezner1984}, and for a convex polygon \cite{shin98}. 

In this paper we consider new versions of the problem where the input consists of 
a set $\mathcal{D}$ of $n$ disks (instead of points):  In the  \emph{intersection problem} we want to
compute two smallest congruent disks
$C_1$ and $C_2$ such that each disk in $ \mathcal{D}$ intersects $C_1$ or
$C_2$, while in the \emph{covering problem}, all disks in $\mathcal{D}$ have 
to be contained in the union of $C_1$ and $C_2$.
To the best of our knowledge these problems have not been considered so far. 
However, linear-time algorithms are known for both the covering and the intersection 
problem with only one disk~\cite{fischer2003,LvK-LargestBounding-10,megiddo89}.

%
%

\paragraph{Our results.} 
In order to solve the intersection problem, we first consider the two-piercing problem: Given 
a set of disks, decide whether there exist two points such that each disk contains at least one 
of these points. We show that this problem can be solved in $O(n^2 \log^2 n)$ expected time 
and $O(n^2\log^2\log\log n)$ deterministic time. Using these algorithms we can solve 
the intersection problem in $O(n^{2}\log^3 n)$ expected time and $O(n^2\log^4 n\log\log n)$ deterministic time.

For the covering problem we consider two cases:
In the \emph{restricted case} each $D\in\mathcal{D}$ has
to be fully covered by one of the disks $C_1$ or $C_2$. In the
\emph{general case} a disk $D\in\mathcal{D}$ can be covered by the
union of $C_1$ and $C_2$. We show how the algorithms for the
intersection problem can be used to solve the restricted covering
case and present an exact algorithm for the general case. We complement these results by giving efficient approximation
algorithms for both cases.

All the results presented in this paper are summarized in the following table.

\begin{center}
\footnotesize
\setlength{\tabcolsep}{5pt}
\renewcommand{\arraystretch}{1.2}
\begin{tabular}{|l| l l|}
\hline
&Exact algorithm&$(1+\epsilon)$-approximation\\ \hline
Intersection problem & $O(n^2\log^4 n \log \log n)$ & --\\
& $O(n^2\log^3 n)$ expected time&\\\hline
General covering problem &$O(n^3\log^4 n)$&$O(n+1/\epsilon^3)$\\ \hline
Restricted covering problem &$O(n^2\log^4 n \log \log n)$ & $O(n+(1/\epsilon^3)\log 1/\epsilon)$\\
&$O(n^2\log^3 n)$ expected time&\\ \hline 
\end{tabular}
\end{center}

\paragraph{Notation.}
The radius of a disk $D$ is denoted by $r(D)$ and its center by
$c(D)$. The circle that forms the boundary of $D$ is denoted by $\partial D$. 

Without loss of generality, we assume that no disk in  $\mathcal{D}$ contains
another disk in $\mathcal{D}$.

\section{Intersecting Disks with Two Centers}\label{sec:Intersecting}
In this section we consider the following
%
%
intersection problem:
Given a set of disks $\mathcal{D}=\{D_1,\dots, D_n\}$, we want to find two 
smallest congruent disks $C_1$ and $C_2$ such that each disk $D\in \mathcal{D}$ has 
a nonempty intersection with $C_1$ or $C_2$.

Based on the observation below, there is an $O(n^3)$ algorithm for
this problem.
\begin{observ}\label{obs:bisector}
Let $(C_1,C_2)$ be a pair of optimal covering disks.
Let $\ell$ be the bisector of the segment connecting the centers of $C_1$ and $C_2$. 
Then, $C_i\cap D\neq\emptyset$ for every $D\in \mathcal{D}$
whose center lies on the same side of $\ell$ as the center of $C_i$,
for $i=\{1,2\}$.
\end{observ}
A simple approach would be, for every bipartition
of the centers of the disks in $\mathcal{D}$ by a line $\ell$, 
to compute the smallest disk intersecting the 
disks 
on each side of $\ell$, 
and return the best result over all bipartitions. Since there are $O(n^2)$ such
partitions, and the smallest disk intersecting a set of disks can be found in
linear time \cite{LvK-LargestBounding-10}, this algorithm runs in $O(n^3)$ time.

We will present faster algorithms for the intersection problem.
We first introduce a related problem. 
For a real number $\delta \geq 0$ and a disk $D$,
the \emph{$\delta$-inflated} disk $D(\delta)$
is a disk concentric to $D$ and
whose radius is $r(D)+\delta$.
Consider the following decision problem:
\begin{quote}
Given a value
$\delta \geq 0$, are there two points $p_1$ and $p_2$
such that $D(\delta)\cap \{p_1,p_2\}\neq\emptyset$
for every $D\in \mathcal{D}$?
\end{quote}
This problem is related to our original problem in the following way. 
The above condition holds with $\delta$ if and only if  
the two disks centered at $p_1$ and $p_2$ with radius $\delta$ intersect 
all disks $D\in \mathcal{D}$.
Therefore the two disks centered at $p_1$ and $p_2$ with radius $\delta^*$ are 
a solution to the intersection problem, where $\delta^*$ is the minimum value
for which the answer to the decision problem is ``yes''.

\subsection{Decision Algorithm}
\label{sec:decision}
Given a value $\delta\geq 0$, we construct the arrangement of the
$\delta$-inflated disks $D_i(\delta), i=1 \dots n$ in the plane. This
arrangement consists
of $O(n^2)$ cells, each cell being a $0, 1$,  or $2$-face.
We traverse all the cells in the arrangement in a depth-first manner
and do the followings:
We place one center point, say $p_1$, in a cell.
The algorithm returns ``yes'' if all
the disks that do not contain $p_1$ have
a nonempty common intersection.
Otherwise, we move $p_1$ to a neighboring cell,
and repeat the test until we visit every cell.
This na\"ive approach leads to a running time $O(n^3)$:
we traverse $O(n^2)$ cells, and each cell can be handled
in linear time.

The following approach allows us to improve this running time by 
almost a linear factor. We consider a traversal of the arrangement of the
$\delta$-inflated disks by a path $\gamma$ that crosses only $O(n^2)$
cells, that is, some cells may be crossed several times, but
on average each cell is crossed $O(1)$ times. It can be
achieved by choosing the path $\gamma$ to be the Eulerian
tour of the depth-first search tree from the na\"ive approach.

While we move the center $p_1$ along $\gamma$ and traverse
the arrangement, we want to know whether the set of disks $\mathcal D'$
that do not contain $p_1$  have a non-empty intersection. To do
this efficiently, we use a segment tree~\cite{bcko-cgaa-08}. 
Each disk of $\mathcal D$
may appear or disappear several times during the traversal of
$\gamma$: each time we cross the boundary of a cell, one disk
is inserted or deleted from $\mathcal D'$. So each disk appears in $\mathcal D'$ along one or several segments of
$\gamma$. We store these segments in a segment tree.
As there are only $O(n^2)$ crossings with cell boundaries along
$\gamma$, this segment tree is built over a total of $O(n^2)$
endpoints and thus has total size $O(n^2 \log n)$: Each segment
of $\gamma$ along which a given disk of $\mathcal D$ is in
$\mathcal D'$ is inserted in $O(\log n)$ nodes of the segment
tree.
Each node $v$ of the segment tree stores a set
$\mathcal D_v \subseteq \mathcal D$ of input disks; from
the discussion above, they represent disks that do not
contain $p_1$ during the whole segment of $\gamma$ that is
represented by $v$. In addition, we store at node $v$
the intersection  $I_v=\bigcap \mathcal{D}_v$
of the disks stored at $v$. Each such intersection $I_v$
is a convex set bounded by $O(n)$ circular arcs, so
we store them as an array of circular arcs sorted along
the boundary of $I_v$. 
In total it takes $O(n^2 \log^2 n)$ time to compute the intersections $I_v$ for all nodes $v$ in the segment tree, since each disk is stored at $O(n \log n)$ nodes on average and the intersection of $k$ disks can be computed in $O(k \log k)$ time.


We now need to decide whether at some point, when $p_1$
moves along $\gamma$, the intersection of the disks in
$\mathcal D'$ (that is, disks that do not contain $p_1$)
is nonempty. To do this, we consider each leaf of the
segment tree separately. At each leaf, we test whether
the intersection of the disks stored at this leaf and all
its ancestors is non-empty. So it reduces to emptiness testing
for a collection of $O(\log n)$ circular
polygons with $O(n)$ circular arcs each.
We can solve this in $O(\log^2 n)$ expected time by
randomized convex programming~\cite{clarkson95,sharirwelzl92}, 
using $O(\log n)$ of the following primitive operations:
\begin{itemize}
\item[1.] Given $I_i, I_j$ and vector $a\in \R^2$, find  
the extreme point $v\in I_i\cap I_j$ that minimizes $a\cdot v$.
\item[2.] Given $I_i$ and a point $p$, decide whether $p\in I_i$.
\end{itemize}
We can also solve this problem in $O(\log^2 n \log \log n)$ time using
deterministic convex programming~\cite{chan98}. So we obtain the
following result:
\begin{lemma}\label{lem:decision_intersect}
  Given a value $\delta\geq 0$, we can decide in $O(n^{2}\log^2 n)$
  expected time or in $O(n^2\log^2 n\log\log n)$ worst-case time
   whether there exist two points such that every
  $\delta$-inflated disk intersects at least one of them.
\end{lemma}


\subsection{Optimization Algorithm}
The following lemma shows that the optimum $\delta^*$ can be found
in a set of $O(n^3)$ possible values. 
\begin{lemma}\label{lem:intersectionOrtangent}
  When $\delta=\delta^*$, $p_1$ or $p_2$ is a common boundary point
  of three $\delta^*$-inflated disks, a tangent point of two
  $\delta^*$-inflated disks or $\delta^*=0$.
\end{lemma}
\begin{proof}
 Suppose that this is not the case. 
 Then the common intersection
 of the disks containing $p_1$ has nonempty interior.
 Similarly, the common intersection of the disks containing $p_2$
 has nonempty interior.
 Let $p'_1$ and $p'_2$ be points in the interiors, one from each
 common intersection. Then there is a value $\delta'<\delta$
 satisfying $D(\delta')\cap \{p'_1,p'_2\}\neq\emptyset$ for every
 $D\in{\mathcal D}$. But we also assumed that $\delta^* \not= 0$.
\end{proof}


\subsection*{Finding $\delta^*$}\label{app:findingdelta}
Due to Lemma~\ref{lem:intersectionOrtangent} we consider only discrete values of $\delta$
for which one of the events defined in Lemma~\ref{lem:intersectionOrtangent}
occurs.
Whether $\delta^*=0$ can be tested  with the decision algorithm in $O(n^{2}\log^2 n)$
  expected time or in $O(n^2\log^2 n\log\log n)$ worst-case time.
So from now on, we assume that $p_1$ or $p_2$ is a common boundary point of three $\delta$-inflated disks or a tangent
point of two $\delta$-inflated disks.


In order to compute all possible values for $\delta$, we
construct a frustum $f_i\in \R^3$ for each disk $D_i\in \mathcal{D}$.
The bottom base of the frustum $f_i$ is $D_i$
lying in the plane $z=0$. The intersection
of $f_i$ and the plane $z=\delta$ is $D_i(\delta)$.
The top base of $f_i$ is $D_i(\delta_\text{max})$, where
$\delta_\text{max}$ is the minimum radius of the disk intersecting all disks in $\mathcal{D}$.
Clearly, the optimal value of $\delta$ is in
$[0, \delta_\text{max}]$.\\

\subsubsection{Event points and their corresponding radii.}
Consider the case that $p_1=(x,y)$ is the common boundary point of the disks
$D_i(\delta)$, $D_j(\delta)$, and $D_k(\delta)$ in the plane.  Then the
point $p'=(x,y,\delta)$ is the common boundary point of three frustums
$f_i$, $f_j$, and $f_k$.  Consider now the case that $p_1=(x,y)$ is
the tangent point of $D_i(\delta)$ and $D_j(\delta)$. Then the point
$p'=(x,y,\delta)$ is the point with the smallest $z$-value on the
intersection curve of $f_i$ and $f_j$.  We call such a point the
\emph{tangent point} of two frustums.  Hence, in order to find the points
$p_1$ and $p_2$, all the tangent points and the common boundary  points of the
frustums have to be considered. There are $O(n^2)$ tangent points and
$O(n^3)$ common boundary points, therefore there are $O(n^3)$ candidates for
the point $p_1$ in total (note that for each candidate for $p_1$, the
corresponding value for $\delta$ is obtained, namely the height of $p_1'$).
Thus, a na\"ive way to find the minimum value $\delta$ such that there exists
two points $p_1,p_2$ that fulfil the conditions, is to test all candidate
$\delta$ values. For each possible $\delta$ value, we can determine
if there are two points $p_1, p_2$ such that all $D(\delta)$ are intersected
by $p_2$ or $p_1$ (as argued above). The solution is the smallest
value $\delta^*$ at which the decision algorithm in
Section~\ref{sec:decision} returns ``yes''.
This leads to
a running time of $O(n^{5}\log^2 n)$ expected time or $O(n^5\log^2 n \log \log n)$ deterministic time.  

In order to improve the running time we use an implicit binary search.

\subsubsection{Implicit binary search.}
We perform an
implicit binary search on the $\delta$ values corresponding to these common boundary points.
As argued above, $p_1$ is the projection of a point $p'$ which is a tangent point of two frustums or a
common boundary point of three frustums, i.e., a vertex of the arrangement
$\mathcal{A}$ of the $n$ frustums $f_1,\dots f_n$; the complexity of
$\mathcal{A}$ is $O(n^3)$.
We now describe how to perform the binary search over the vertices of
$\mathcal{A}$ in an implicit way:

\paragraph{Binary search on a coarse list of events.}
We first 
consider $O(n^2)$ pairs of frustums and compute the tangent point
of each pair. 
Then we randomly select $O(n^2 \log n)$ triples of frustums and compute 
the common boundary point of each triple. Since $\delta^* \in [0, \delta_\text{max}]$, we only consider points whose $z$-value is in this interval. Clearly, these points are vertices of $\mathcal{A}$ and hence we randomly select $O(n^2 \log n)$ vertices from $\mathcal{A}$. We
sort their radii associated with them in $O(n^2 \log^2 n)$ time. By a binary
search with the decision algorithm in Section~\ref{sec:decision},
we determine two consecutive radii
$\delta_i$ and $\delta_{i+1}$ such that $\delta^*$ is between
$\delta_i$ and $\delta_{i+1}$. This takes $O(n^{2}\log^3 n)$ time. Since the
vertices were picked randomly, the strip $W[\delta_i, \delta_{i+1}]$
bounded by the two planes $z:= \delta_i$ and $z:= \delta_{i+1}$ contains
only $k= O(n)$ vertices of
$\mathcal{A}$ with high probability \cite[Section~5]{mulmuley1994}.

\paragraph{Zooming into the interval.}
We compute all
the $k$ vertices in $W[\delta_i,\delta_{i+1}]$ by a standard sweep-plane
algorithm in $O(k\log n + n^2\log n)$ time as follows: First, we compute
the intersection of the sweeping plane at $z:= \delta_i$ with the frustums
$f_1,\dots, f_n$. This intersection forms a two-dimensional arrangement of $O(n)$
circles with $O(n^2)$ total complexity, and we can compute it in $O(n^2\log
n)$ time. We next construct the portion of the arrangement $\mathcal{A}$ in
$W[\delta_i,\delta_{i+1}]$ incrementally by sweeping a plane orthogonal to 
$z$-axis from
the intersection at $z:=\delta_i$ towards $z:= \delta_{i+1}$. As a result,
we can compute the $k=O(n)$ vertices (and the corresponding $O(n)$ radii)
in $W[\delta_i,\delta_{i+1}]$ in $O(n\log n)$ time.  We abort the sweep if
the number $k$ of vertices inside the strip becomes too large and restart
the algorithm with a new random sample. This happens only with small
probability.  In order to find the minimum value $\delta^*$,
we perform a binary search on these $O(n)$ radii we just computed,
using the decision algorithm in Lemma~\ref{lem:decision_intersect}.
This takes $O(n^{2}\log^3 n)$ expected
time. The solution pair of points $p_1$ and $p_2$ can also be found
by the decision algorithm.
\medskip

To get a deterministic algorithm, we use the parametric search technique, with the deterministic decision 
algorithm of Lemma~\ref{lem:decision_intersect}. As the generic algorithm, we use an algorithm that computes 
in $O(\log n)$ time the arrangement of the inflated disks using $O(n^2)$ processors~\cite{agarwal1994}, so 
we need to run the decision algorithm $O(\log^2 n)$ times, and
the total running time becomes $O(n^2 \log^4 n\log\log n)$.
\begin{theorem}
  Given a set $\mathcal{D}$ of $n$ disks in the plane, we can compute
  two smallest congruent disks whose union intersects every disk in
  $\mathcal{D}$ in $O(n^{2} \log^3 n)$ expected time, and 
  in $O(n^2 \log^4 n \log\log n)$ deterministic time.
\end{theorem}

\section{Covering Disks with Two Centers}
In this section we consider the following 
covering problem:
Given a set of disks $\mathcal{D}=\{D_1, \dots, D_n\}$, compute two smallest congruent disks $C_1$ and $C_2$ such that each disk $D\in \mathcal{D}$ is covered by $C_1$ or $C_2$.
In the \emph{general  case}, a disk $D\in\mathcal{D}$ must be covered by $C_1\cup C_2$.
In the \emph{restricted  case}, each disk $D\in\mathcal{D}$ has to be fully covered by $C_1$ or by $C_2$.

\subsection{The General Case}

We first give a characterization of the optimal covering.
The optimal covering of a set $\mathcal D'$ of disks by one disk is determined by
at most three disks of $\mathcal D'$ touching the optimal covering disk such that the convex hull of the contact points contains the center of the 
covering disk. (See Figure~\ref{fig:GeneralOptimal}(a).) 
\begin{figure}[ht]
\centering
 \includegraphics[scale=0.6]{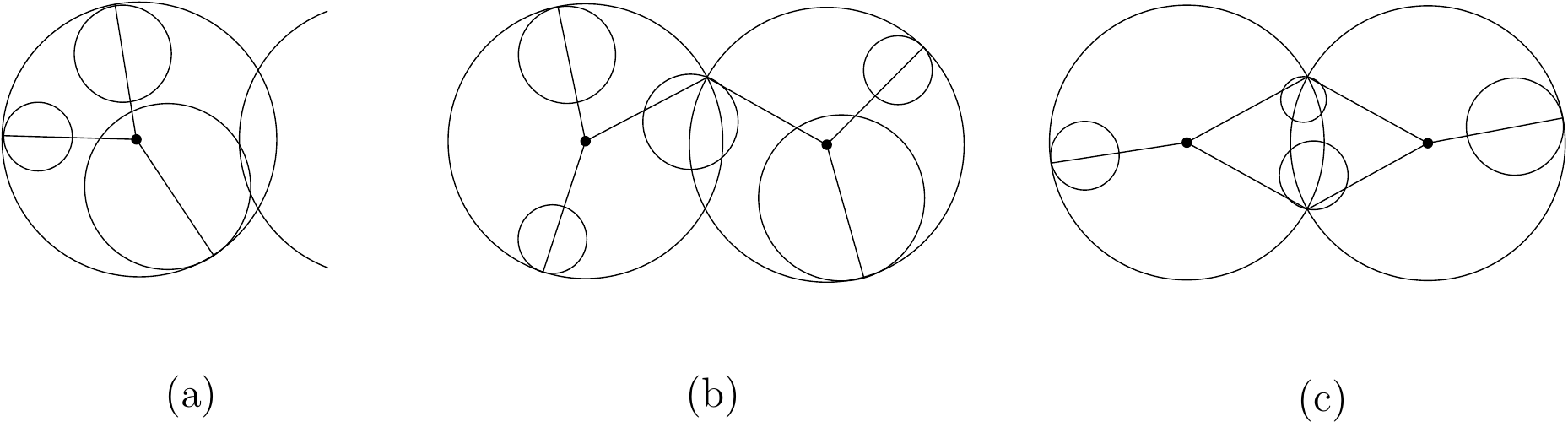}
\caption{The three configurations for the optimal 2-center covering of disks.}\label{fig:GeneralOptimal}
\end{figure}

When covering by two disks, a similar argument applies,
and thus the optimal covering disks $(C_1^*, C_2^*)$ are determined by at most
five input disks. 
\begin{lemma}\label{lem:generalOptimal}
The optimal covering by two disks $C_1^*, C_2^*$ satisfies one of
the following conditions.
\begin{enumerate}
\item For some $i \in \{1,2\}$, the disk $C_i^*$ is the optimal one-covering  of the
	disks contained in $C^* _i$, as in
  Figure~\ref{fig:GeneralOptimal}(a).
\item There is an input disk that is neither fully contained in $C_1^*$ nor
  in $C_2^*$, but contains one point of $\partial C_1^*\cap \partial C_2^*$ in its boundary as in
  Figure~\ref{fig:GeneralOptimal}(b).
\item There are two input disks $D_i, D_j$, possible $i=j$, none of them being fully covered by $C_1^*$ or $C_2^*$, such that $D_i$ contains one point of $\partial C_1^*\cap \partial C_2^*$ and $D_j$  contains the other 
  point of $\partial C_1^*\cap \partial C_2^*$ in their boundaries as in
  Figure~\ref{fig:GeneralOptimal}(c).
\end{enumerate}
In all cases, each covering disk $C^*$ is determined by at most three 
disks whose contact points contain the center $c(C^*)$ in their convex hull.
\end{lemma}
\begin{proof}
The optimal solution is a pair of congruent disks that achieves
a local minimum in radius, that is, we cannot
reduce the radius of the covering disks by translating them locally.
If one covering disk is completely determined by the input disks contained
in it, then this belongs to case 1.
Otherwise, there always exists at least one input disk $D$ such that
$D$ is not contained in $C_i^*$ for all $i\in\{1,2\}$. Moreover such input disks always
touch $C_1^*\cup C_2^*$ from inside at the intersection points of
$\partial C_1^*$ and $\partial C_2^*$, otherwise we can always get a pair 
of smaller congruent covering disks.
If only one point of $\partial C_1^*\cap\partial C_2^*$ is touched by
an input disk $D$,  both covering disks are determined by at most two
additional disks touching from
inside
together with $D$ because the covering disks are congruent.
If both intersection points of $\partial C_1^*\cap\partial C_2^*$ are touched by input disks $D_i$ and $D_j$, possible $i=j$, 
one covering disk is determined by one additional
disk and the other covering disk by at most one additional disk
touching from inside together with $D_i$ and $D_j$
because the covering disks are congruent.
It is not difficult to see that there are two or three touching points of each covering
disk that make radial angles at most $\pi$; otherwise we can get a pair of
smaller congruent covering disks.
\end{proof}

Using a decision algorithm and the parametric search technique, we can construct an exact algorithm for the general covering problem. 

Let $r^*$ be the radius of an optimal solution for the general case of covering by two disks.
We describe a decision algorithm based on the following lemma that, for a given $r>0$, returns
``yes'' if $r\geq r^*$, and ``no'' otherwise. (See also Figure~\ref{fig:GeneralCases}).

\begin{lemma}
\label{lem:general.cases}
Assume that $r\geq r^*$. Then there exists
a pair of congruent disks $C_1,C_2$ of radius $r$ such that
their union contains the input disks, an input disk $D$ touches $C_1$
from inside, and one of the following property holds.
\begin{itemize}
\item[(a)] $C_1$ is identical to $D$.
\item[(b)] There is another input disk touching $C_1$ from inside.
\item[(c)] There is another input disk $D'$ such that $D'$ is not contained
  in $C_2$, but it touches a common intersection $t$ of $\partial C_1$ and $\partial C_2$
  that is at distance $2r$ from the touching
  point of $D$. If this is the case, we say that $D$ and $t$ are \emph{aligned
  with respect to $C_1$}.
\item[(d)] There are two disks $D_i$ and $D_j$, possibly $i=j$, such that
  $D_i$ touches a common intersection of $\partial C_1$ and $\partial C_2$,
  and $D_j$ touches the other common intersection of
  $\partial C_1$ and $\partial C_2$.
\end{itemize}
\end{lemma}
\begin{proof}
Let $c_1^*$ and $c_2^*$ be the centers of the optimal solution. Imagine that
we place two disks at $c_1^*$ and $c_2^*$ with radius larger than $r^*$, respectively.
If $C_1$ is already identical to an input disk $D$, it belongs to case (a).
Otherwise we translate $C_1$ towards $C_2$ until it hits an input disk $D$.
Then we rotate $C_1$ around $D$ in clockwise orientation
maintaining $D$ touching $C_1$ from inside
until the union of $C_1$ and $C_2$ stops covering the input.
If this event is caused by another disk touching $C_1$ from inside,
it belongs to case (b).
Otherwise the event is caused by another disk $D_i$ that is hit
by one of two common intersections of $\partial C_1$ and $\partial C_2$
at $t$. If $D$ and $t$ are aligned with respected to $C_1$,
it belongs to case (c).

Otherwise, we rotate $C_2$ around $t$ in counterclockwise until the
union of $C_1$
and $C_2$ stops covering the input. If this event is caused by
another disk $D_j$ that is hit by the common intersection of
$\partial C_1$ and $\partial C_2$, other than $t$, then it belongs to case (d). Otherwise the event is caused by another disk $D'$
touching $C_2$ from inside. Thus, $D$ touches $C_1$ from inside,
$D'$ touches $C_2$ from inside, and $D_i$ touches the common intersection
$t$ of $\partial C_1$ and $\partial C_2$. Imagine that we
rotate $C_1$ slightly further around $D$ in clockwise.
We also rotate $C_2$ around $D'$ simultaneously such that
$D_i$ and the rotated copies keep maintaining a common intersection
along their boundaries during the rotation. 
Let $t$ denote the common intersection.
We rotate $C_1$ and $C_2$ in such a way until we encounter an event
(1) that another disk $D_j$ touches $C_1$ or $C_2$, (2) that $D_j$ touches
the other common intersection of $\partial C_1$ and $\partial C_2$,
or (3) that $D$ and $t$ are aligned with respect to $C_1$ or
$D'$ and $t$ are aligned with respect to $C_2$.
Note that if the event is of type (3), then $D_i$ is not
contained in the disk centered at $c(D)$ with radius $2r-r(D)$ or
is not contained in the disk centered at $c(D')$ with radius
$2r-r(D')$ as in Figure~\ref{fig:GeneralCases}. 
\end{proof}
\begin{figure}[ht]
\centering
 \includegraphics[scale=0.8]{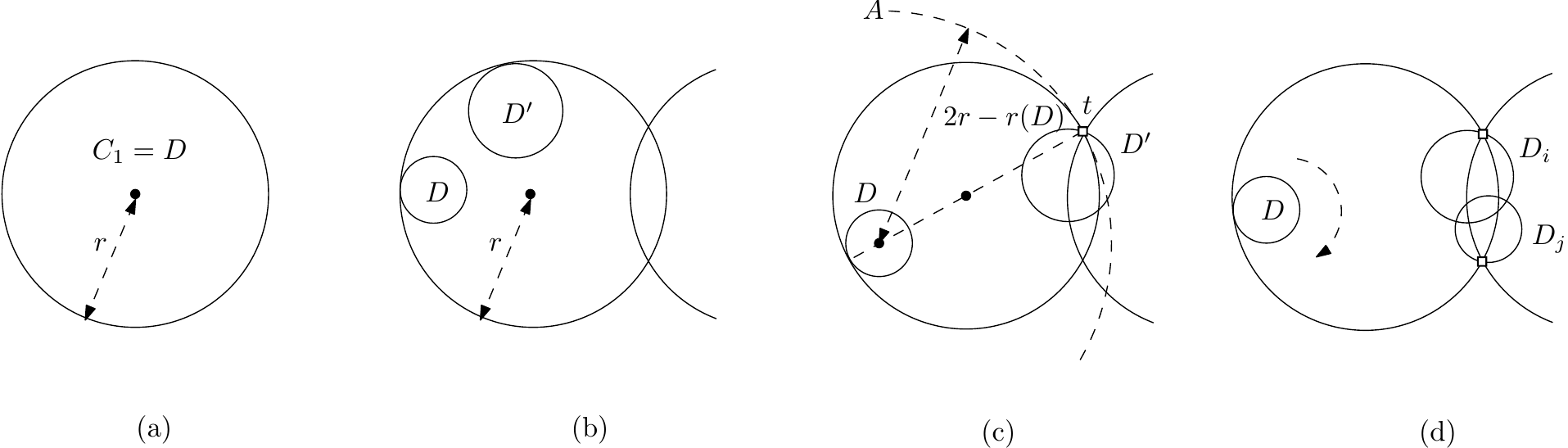}
\caption{Four cases for $r\geq r^*$.}\label{fig:GeneralCases}
\end{figure}

\subsubsection{Decision  Algorithm.}\label{ssec:generalDecision}
The cases are enumerated as in Lemma~\ref{lem:general.cases}.\\
\emph{Case (a).} Choose an input disk $D$. $C_1$ has radius $r$ and covers only $D$. Then $C_2$ is the smallest disk containing $\mathcal{D} \setminus D$. If the radius of $C_2$ is $\leq r$, we return ``yes".
\\
\emph{Case (b).} We simply choose a pair of input disks $D$ and $D'$. There are two candidates for $C_1$, 
as $C_1$ has radius $r$ and touches $D$ and $D'$. So we consider separately each of the two candidates for $C_1$. 
Then $C_2$ is chosen to be the smallest disk containing the input disks, or the 
portions of input disks (crescents) that are not covered by $C_1$, which can be
done in $O(n)$ time. If for one of the two choices of $C_1$, the corresponding
disk $C_2$ has radius $\leq r$, we return ``yes''.
\\
\emph{Case (c).} For each input disk $D$, we do the following.
\begin{enumerate}
\item For the circle $A$ with center $c(D)$ and radius $2r-r(D)$, compute $A\cap D'$ for every other disk $D'$. Let $t$ be such an intersection point.
\item For each $t$,
\begin{enumerate}
\item remove (part of) the input disks covered by the covering
disk determined by $D$ and $t$, and compute the smallest disk covering the remaining input.
\item If this algorithm returns a covering disk with radius $\leq r$, return ``yes''.
\end{enumerate}
\end{enumerate}

\emph{Case (d).} For each input disk $D$ that touches $C_1$ from
inside, we do the following.
Let $i$ be the index of the first input disk that the
circular arc of $C_1$ from the touching point hits in clockwise
orientation. Let $j$ be the index of the last input disk that the circular
arc leaves.
We claim that the number of pairs of type $(i,j)$ is $O(n)$.

This claim can be easily proved by observing that,
while we rotate $C_1$ around an input disk $D$
in clockwise orientation,
$C_1$ sweeps the plane and the input disks in such a manner
that the first input disk intersected by the arc of $C_1$ from
the tangent point in clockwise orientation changes only $O(n)$ times;
To see this, consider the union of the input
disks, which consists of $O(n)$ circular arcs.
The last input disk intersected also changes $O(n)$ times. 
So the pairing along the rotation can be
done by scanning two lists (first and last) of disks.
For each pair $(i,j)$, we still have some freedom of rotating $C_1$ around $D$ within some interval ($C_2$ changes accordingly.) During the rotation, an input disk
not covered by the union of $C_1$ and $C_2$ may become fully covered by
the union,
or vice versa. We call such an event an \emph{I/O event}. Note that an I/O event occurs only when an input disk touches $C_1$ or $C_2$ from inside.
Again, we claim that the number of I/O events for each pair $(i,j)$ is $O(n)$.

For this claim, consider a pair $(i,j)$.
During the rotation, the first intersection point moves along
the boundary of disk $D_i$ and the last intersection point moves
along the boundary of disk $D_j$.
Therefore, the movement of $C_2$
is determined by these two intersection points.
Clearly $C_1$ has at most $2(n-1)$ I/O events. For $C_2$, the trajectory
of its center is a function graph which ``behaves well'' --
Since it is a function on the radii of disks $D_i$ and $D_j$, and their
center locations, it is not in a complicated form (and its
degree is low enough) that there are only $O(n)$ events.

We compute all I/O  events and sort them. At the beginning of the rotation of $C_1$ around $D$, 
we compute the number of input disks that are not fully covered, and set the variable \emph{counter} 
to this number. Then we handle I/O events one by one and update the counter. 
If the counter becomes 0, we return ``yes''.

In total, case (d) can be handled in $O(n^3\log n)$ time.

\begin{lemma}
Given a value $r>0$, we can decide in $O(n^3\log n)$ time whether there exists two disks with radius $r$ that cover a set of given disks in the plane.
\end{lemma}
For the optimization algorithm we use parametric search.

To use the parametric search technique, we will design a parallel version
of the decision algorithm. Then the overall algorithm runs in time
$O(p\cdot T_p+T_p\cdot T_d\log p)$, where $p$ denotes the number of
processors, $T_d$ denotes the running time of a decision algorithm,
and $T_p$ denotes the running time of the parallel decision
algorithm using $p$ processors. We have a parallel decision algorithm where
$p = O(n^3)$ processors and $T_p = O(\log^2 n)$ time for some constant $c> 1$. Thus the overall algorithm runs in time $O(n^3\log^{4}n)$ time.

\paragraph{Parallel decision algorithm.}
For case (a), we assign $O(n)$ processors to each candidate $D$.  The 1-center disk covering for the disks in $\mathcal{D}\setminus D$ can be computed by a known parallel linear programming algorithm~\cite{goodrich96} in $O(\log^2 n)$ time with $O(n)$ processors.
Hence, we can solve case (a) in $O(\log^2 n)$ time with $O(n^2)$ processor.

For case (b), we assign $O(n)$ processors to each pair $(D, D')$ of input disks. With a covering disk $C_1$ of radius $r$ determined by $D$ and $D'$, we cover the input disks and compute the crescents of the input disks not covered by $C_1$ in a constant time. The 1-center disk covering for the crescents can be computed in $O(\log^2 n)$ time with $O(n)$ processors~\cite{goodrich96}. Thus we can handle the case (b) in $O(\log^2 n)$ time with $O(n^3)$ processors, and moreover we can deal with case (c) in a similar way.

For case (d), we first compute the union $U$ of all input disks in $O(\log n)$ time with $O(n^3)$ processors~\cite{papa87}. Next we assign $O(n^2)$ processors to each input disk $D$. We fix $D$. As $C_1$ rotates around $D$ while they keep touching as in Figure~\ref{fig:GeneralCases}(c), we need to figure out $O(n)$ pairs $(i, j)$ of input disks such that $D_i$ and $D_j$ are the first and the last ones intersected by $C_1$, respectively. Such $D_i$ and $D_j$ must be on the boundary of $U$, so it is sufficient to consider the disks whose arcs appear on the boundary of $U$. To get these pairs, we assign $O(n)$ processors to each disk on the union boundary in order to calculate two rotating angles of $C_1$ at which $C_1$ hits the disk at the first and the last in $O(1)$ time. We collect all these angles, sort them, and extract the pairs $(i, j)$ from the sorted list; all steps are easily done in $O(\log n)$ time using $O(n)$ processors.

For a fixed angle interval $I$ determined by some pair $(i, j)$, the set of input disks not covered by $C_1$ remains same, and we can also know which disks are those ones. Using $O(n)$ processors in $O(1)$ time for each input disk $D'$ not covered by $C_1$, we compute the subintervals $J \subseteq I$ such that $C_1\cup C_2$ determined by $J$ contains $D'$ at any angle in $J$. These subintervals are defined by I/O events we mentioned in the sequential decision algorithm, so there are $O(n)$ subintervals. Finally we test whether the intersection of the subintervals is empty or not. If it is not empty, then it means there is a rotation angle in $I$ at which all input disks not covered by $C_1$ get to be contained in $C_1\cup C_2$. Otherwise, no angles in $I$ guarantee the full coverage by $C_1\cup C_2$. This test can be done in bottom-up fashion in $O(\log n)$ time using $O(n)$ processors. After testing all pairs $(i, j)$, if there is a pair such that the intersection is not empty, then return ``yes''. This is done in $O(\log n)$ time with $O(n^2)$ processors for a fixed disk $d$ touching $C_1$ from inside. Summing up all things, we can solve case (d) in $O(\log n)$ time with $O(n^3)$ processors.

\begin{theorem}
Given a set of $n$ disks in the plane, we can find a pair of congruent disks
with smallest radius whose union covers all of them in $O(n^3\log^4 n)$ time.
\end{theorem}

\subsubsection{Constant Factor Approximation.}\label{ssec:generalCase}

We apply the well known greedy $k$-center approximation algorithm by
Gonzalez~\cite{Gonzalez85} to our general covering case.
It works as follows: First pick an arbitrary point $c_1$ in the union $\bigcup \mathcal D$
of our input disks. For instance, we could choose $c_1$ to be the center of $D_1$. Then
compute a point $c_2 \in \bigcup \mathcal D$ that is farthest from $c_1$.
This can be done in linear time by brute force. These two points are the centers of our
two covering disks, and we choose their radius to be as small as possible, that is,
the radius of the two covering disks is the maximum distance from any point in
$\bigcup \mathcal D$ to its closest point in $\{c_1,c_2\}$.
This algorithm is a  2-approximation algorithm, so we obtain the following result:
\begin{theorem}\label{th:generalconstant}
We can compute in $O(n)$ time a $2$-approximation for the
general covering problem for a set $\mathcal{D}$ of $n$ disks.
\end{theorem}


\subsubsection{$(1+\epsilon)$-Approximation.}

Our $(1+\epsilon)$-approximation algorithm is an adaptation of an algorithm
by Agarwal and Procopiuc~\cite{Agarwal02}.
We start by computing a $2$-approximation for the general covering case in
$O(n)$ time using our algorithm from Theorem~\ref{th:generalconstant}.
Let $C_1$, $C_2$ be the disks computed by this approximation algorithm and let $r$ be their
radius.
We consider a grid of size $\delta=\lambda \epsilon r$ over the plane, where
$\lambda$ is a small enough constant. That is, we consider the  points with coordinates
$(i\delta,j\delta)$ for some integers $i,j$. Observe that there are only $O(1/\epsilon^2)$
grid points in $C_1 \cup C_2$.
The center of each disk $D$ is moved to a nearby grid point. That is, a center $(x,y)$
is replaced by  $(\delta \lceil x/\delta \rceil, \delta \lceil y/\delta \rceil)$. If two
or more centers are moved to the the same grid point, we only keep the disk with the largest radius.
All the centers are now grid points inside $C_1 \cup C_2$, or at distance at most $\sqrt 2 \delta$
from the boundary of this union, so we are left with a set of $O(1/\epsilon^2)$ disks.
We now replace this new set of disks by grid points: each disk is replaced by the grid points
 which are
closest to the boundary of this disk and lie inside this disk.
In order to compute these points we consider each column of the grid separately: The intersection
of each disk with this column is an interval, and we replace the interval by the lowest  and
the highest grid point lying inside this interval.  Since the set of disks has size $O(1/\epsilon^2)$
and the number of columns is $O(1/\epsilon)$, it takes in total $O(1/\epsilon^3)$ time.
The set of grid points we obtain is denoted by $P_g$ and its size is $O(1/\epsilon^2)$.
We compute two smallest disks $E_1,E_2$ that cover $P_g$ in
$O(\frac{1}{\epsilon^2}\log^2 \frac{1}{\epsilon}\log^2 \log \frac{1}{\epsilon})$ time
using the algorithm from Chan~\cite{chan99}.
Choosing the constant $\lambda$ small enough and increasing the radii of $E_1,E_2$ by $2\sqrt{2}\delta$, these disks
are a $(1+\epsilon)$-approximation
of the solution to our general disk cover problem.

\begin{theorem}\label{th:generalLTAS}
 Given a set $\mathcal{D}$ of $n$ disks in the plane, a
  $(1+\epsilon)$-approximation for $\mathcal{D}$ in the general covering case can be computed in $O(n+1/ \epsilon^3)$
  time.
\end{theorem}

\subsection{The Restricted Case}

Observation~\ref{obs:bisector} can be adapted to the restricted covering case.
\begin{observ}\label{obs:bisectorCovering}
Let $\ell$ be the bisector of an optimal solution $C_1$ and
$C_2$. Then, $D\subset C_i$ for every $D\in \mathcal{D}$
whose center lies in the same side of $\ell$ as the center of $C_i$,
for $i=\{1,2\}$.
\end{observ}
Hence, the restricted covering problem can be solved in $O(n^3)$ time, since for a set  of $n$ disks $\mathcal{D}$  the smallest disk covering all $D\in\mathcal{D}$ can be computed in $O(n)$ time \cite{megiddo89} and there are $O(n^2)$ different bipartitions of the centers of the disks.

The algorithm from Section~\ref{sec:Intersecting} can also be adapted to solve the restricted covering problem.
We consider the decision problem, which can be formulated as follows:
Given a set of $n$ disks $\mathcal{D}$ and a value $\delta$, we want to decide whether there exists two disks $C_1, C_2$ with radius $\delta$, such that each disk $D_i\in \mathcal{D}$ is covered by either $C_1$ or $C_2$.
This implies that for each disk $D_j\in \mathcal{D}$ covered by $C_i$, the following holds: $d(c(D_j),c(C_i))+r(D_j)\leq \delta$, for $i=\{1,2\}$.
Let $r_{\text{max}}$ be the maximum of radii of all disks in $\mathcal{D}$. It holds that $\delta\geq r_{\text{max}}$, since if $\delta<r_{\text{max}}$ there clearly exists no two disks with radius $\delta$ which cover $\mathcal{D}$.
We can formulated the problem in a different way.

\begin{quote}
Given a value
$\delta$, do there exist two points, $p_1$ and $p_2$,
in the plane such that $D^*({\delta})\cap \{p_1,p_2\}\neq\emptyset$
for every $D\in \mathcal{D}$, where $D^*({\delta})$ is a disk concentric to $D$ and
whose radius is $\delta-r(D)\geq 0$.
\end{quote}
Recall the definition of $\delta$-\emph{inflated} disks from Section~\ref{sec:Intersecting}. Every disk $D\in \mathcal{D}$ was replaced by a disk concentric to $D$ and whose radius was $r(D)+\delta$.
Here we actually need to replace each disk $D$ by a disk that is concentric to $D$ and has a radius $\delta-r(D)$. Since we know that $\delta\geq r_{\text{max}}$, we add an initialization step, in which every disk $D$ is replaced by a disk concentric to $D$ and
whose radius is $r_{\text{max}}-r(D)$. Then we can use exactly the same algorithm in Section~\ref{sec:Intersecting} in order to compute a solution for the restricted covering problem. Let $\delta^*$ be the solution value computed by this algorithm. Clearly the solution for the covering problem is then $\delta^*+r_\text{max}$.
We summarize this result in the following theorem.
\begin{theorem}
  Given a set of $n$ disks $\mathcal{D}$ in the plane, we can compute
  two smallest congruent disks such that each disk in $\mathcal{D}$ is
  covered by one of the disks in $O(n^{2}\log^3 n)$ expected time or in $O(n^2\log^4 n \log\log n)$ worst-case time.
\end{theorem}

\subsubsection{Constant Factor Approximation.}
Let $C_1,C_2$ denote an optimal solution to the general case, and let $r_g$ be their radius. Then any
solution to the restricted case is also a solution to the general case, so we have $r_g$ is at most
the radius of the optimal solution to the restricted case. On the other hand, the inflated disks
$C_1(2r_g),C_2(2r_g)$ form a solution to the restricted case, because any disk contained
in $C_1 \cup C_2$ should be contained in either $C_1(2r_g)$ or $C_2(2r_g)$.
So we obtain a 6-approximation algorithm for the restricted case by first applying our 2-approximation
algorithm for the general case (Theorem~\ref{th:generalconstant}) and then multiplying by 3 the radius
of the two output disks:
\begin{theorem}\label{th:restrictedconstant}
  Given a set of $n$ disks $\mathcal{D}$ in the plane, we can compute in
  $O(n)$ time a $6$-approximation to the restricted covering problem.
\end{theorem}
As in the general case, we will see below how to improve it to a linear time algorithm for any
constant approximation factor larger than 1.

\subsubsection{$(1+\epsilon)$-Approximation.}

Recall Observation~\ref{obs:bisectorCovering}. Let $\ell$ be the bisector of an optimal solution. Then each disk $D\in\mathcal{D}$ is covered by the disk $C_i$ whose center lies in the same side of the center of $C_i$, $i\in\{1,2\}$.
Hence, if we know the bisector, we know the bipartition of the disks. First, we show how to compute an optimal solution in $O(n\log n)$ time if the direction of the bisector is known. Later on we explain how this algorithm is used in order to obtain a
 $(1+\epsilon)$ approximation.

\paragraph{Fixed Orientation.}
W.l.o.g, assume that the bisector is vertical.
After sorting the centers of all $D \in \mathcal{D}$
by their $x$-values,
we sweep a vertical line $\ell$ from left to right, and
maintain two sets $\mathcal{D}_1$ and $\mathcal{D}_2$:
$\mathcal{D}_1$ contains all disks whose centers lie to the left of
$\ell$ and $\mathcal{D}_2=\mathcal{D}\setminus\mathcal{D}_1$.
Let $C_1$ be the smallest disk covering $\mathcal{D}_1$ and $C_2$ the smallest
disk covering $\mathcal{D}_2$. While sweeping $\ell$ from left to right,
the radius of $C_1$ is nondecreasing and the radius of $C_2$ nonincreasing
and we want to compute $\min \max(r(C_1),r(C_2))$. Hence, we can perform a binary search on the list of the centers of the disks in $\mathcal{D}$. Each step takes $O(n)$ time, thus we achieve a total running time of $O(n \log n)$.

\paragraph{Sampling.}
We use $2\pi / \epsilon$ sample orientations chosen regularly over
$2\pi$, and compute for each orientation the solution in $O(n\log n)$
time.  The approximation factor can be proven by showing that there is
a sample orientation that makes angle at most $\epsilon$ with the
optimal bisector. Without loss of generality,
we assume that the bisector, denoted by $b$, of an optimal solution
$C_1^*$ and $C_2^*$ is vertical as in Figure~\ref{fig:EpsilonApprox}. 
Let $q$ denote the midpoint of the segment connecting 
$c(q_1)$ and $c(q_2)$.

Let $\ell$ be the line
which passes through $q$ and
makes angle with $b$ is at most $\epsilon$ in counterclockwise
direction as in the figure.  (For simplification we set the angle in
the calculation to exactly $\epsilon$.)
\begin{figure}[ht!]
  \centering
  \includegraphics[scale=0.6, page=2]{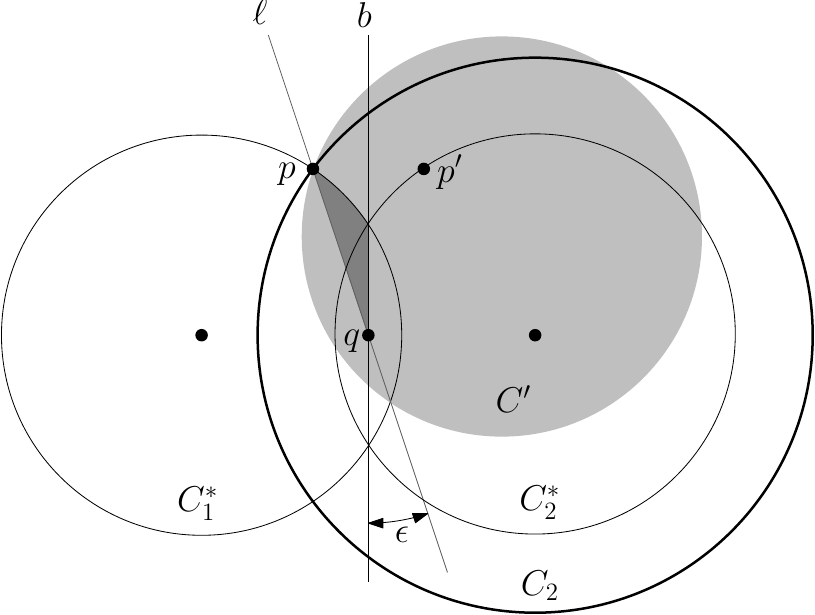}
  \caption{$r(C_2)\leq (1+\epsilon')r(C^*_2)$ for any $\epsilon'\geq 4\epsilon$.}\label{fig:EpsilonApprox}
\end{figure}
Let $p$ denote the intersection point of $\ell$ with the upper circular
arc of $\partial C_1^*$, and let $p'$ denote the point symmetric to $p$ 
along $b$. Clearly $p'$ lies on the boundary of $C_2^*$. 
We will
show that there exists two disks $C_1,C_2$ where $C_1$ covers all
disks whose centers lie to the left of $\ell$ and $C_2$ covers all
disks whose centers lie to the right of $\ell$ and $r(C_1)=r(C_2)\leq
(1+\epsilon) r(C_1^*)$.

We will explain the construction of $C_2$ and prove that $C_2$ covers 
all disks whose centers lie to the right of $\ell$. $C_1$ can be 
constructed analogously. 
The center of $C_2$ is set to $c(C_2^*)$ and the radius is set to
$|c(C^*_2)p|\leq |\overline{c(C_2^*)p'}|+|\overline{p'p}|$. It holds that
$|\overline{c(C_2^*)p'}|+|\overline{p'p}|\leq
r(C_2^*)+4r(C_2^*)\sin\epsilon$, since $|\overline{p'q}|\leq 2r(C_2^*)$ and
the distance of $p'$ to $b$ is at most $2r(C_2^*)\sin\epsilon$.
Clearly $C_2$ covers all disks that were covered by
$C_2^*$. In addition, it must cover all disks whose centers lie in
the region of $C^*_1$ that is bounded by $\ell$ and $b$ and that has $q$ as
its lowest point, depicted as the dark gray region in
Figure~\ref{fig:EpsilonApprox}. Note that the disks whose centers lie
in this region are fully covered by $C^*_1$, but not necessarily by $C^*_2$.
 
It remains to prove that all disks
having their center in the dark gray region are fully covered by $C_2$.
Let $C'$ be the disk symmetric to $C^*_1$ along $\ell$.
Then all disks whose centers lie in the dark gray region are covered by 
$C^*_1\cap C'$, because this region is symmetric along $\ell$ and 
they are fully covered by $C^*_1$. 
Since $C_2$ contains the intersection $C^*_1\cap C'$, we conclude that 
all disks whose centers lie
on the right side of $\ell$ are covered by $C_2$.

We can prove the analog for $C_1$.
Hence, $$r(C_1)=r(C_2)\leq (1+4\sin\epsilon) r(C_1^*)\leq (1+\epsilon')r(C_1^*)=(1+\epsilon')r(C_2^*)$$
as $\sin\epsilon\leq \epsilon$ for $\epsilon\leq 1$ (can be shown by using the theory of Taylor series) and for any $\epsilon'\geq 4\epsilon$.
Since any solution whose bisector is parallel to $\ell$ has a radius 
at most $r(C_1)$, this solution has radius at most $(1+\epsilon')$ times 
the optimal radius.

\begin{theorem}\label{th:generalFPTAS}
 For a given a set $\mathcal{D}$ of $n$ disks in the plane, a
  $(1+\epsilon)$ approximation for the restricted covering problem for $\mathcal{D}$ can be computed in
$O( (n/\epsilon)\log n)$ time.
\end{theorem}
The running time can be improved to $O(n+1/\epsilon^3  \log 1/\epsilon)$ in the following way.
We start with computing a $6$-approximation in $O(n)$ time, using Theorem~\ref{th:restrictedconstant}.
Let $C'_1$ and $C'_2$ be the resulting disks, and let $r'$ be their radius. As in the
proof of Theorem~\ref{th:generalLTAS}, we round
the centers of all input disks $D\in\mathcal{D}$ to grid points inside $C'_1 \cup C'_2$, with a grid
size $\delta'=\lambda' \epsilon r'$, for some small enough constant $\lambda'$.
Then we apply our FPTAS from Theorem~\ref{th:generalFPTAS} to this set of rounded disks and inflate the resulting disks by a factor of $\sqrt{2}\delta$. These disks are a $(1+\epsilon)$-approximation for the optimal solution. As there are
only $O(1/\epsilon^2)$ rounded disks, this can be done in $O((1/\epsilon^3) \log 1/\epsilon)$ time.

\begin{theorem}
 For a given a set $\mathcal{D}$ of $n$ disks in the plane, a
  $(1+\epsilon)$ approximation for the restricted covering problem
for $\mathcal{D}$ can be computed in $O( n+ (1/\epsilon^3)  \log 1/\epsilon)$ time.
\end{theorem}

\subsection*{Acknowledgment}
Work by Ahn was supported by the National Research Foundation of
Korea Grant funded
by the Korean Government (MEST) (NRF-2010-0009857).
    Work by
    Schlipf was supported by the German Science Foundation (DFG)
    within the research training group 'Methods for Discrete
    Structures'(GRK 1408).
Work by Shin was supported by the National Research Foundation of
Korea Grant funded
by the Korean Government (MEST) (NRF-2011-0002827).

{\small \bibliographystyle{abbrv} \bibliography{paper}}

\end{document}